\documentclass[letterpaper, 10 pt, conference]{ieeeconf} 
\IEEEoverridecommandlockouts                             
 \usepackage{graphicx}
\usepackage{textcomp}
 \usepackage{xcolor}
\usepackage{amsfonts}
\usepackage{graphicx}
 \usepackage{amsmath}
 \usepackage{amssymb}
\usepackage{algorithm}
\usepackage[noend]{algpseudocode}
\usepackage{booktabs}
\usepackage{color} 
\usepackage{cite}
\newtheorem{assumption}{Assumption}
\newtheorem{theorem}{Theorem}
\newtheorem{remark}{Remark}
\newtheorem{problem}{Problem}
\usepackage{graphics}
 \usepackage{amsmath}
 \usepackage{amssymb}
\newtheorem{lemma}{Lemma}

\newtheorem{definition}{Definition}

\begin{document}
\title{Adaptive Dynamic Programming and Data-Driven Cooperative Optimal Output Regulation with Adaptive Observers}


\author{Omar Qasem,~\IEEEmembership{Graduate~Student~Member,~IEEE}, Khalid Jebari and Weinan Gao,~\IEEEmembership{Senior Member,~IEEE}
\thanks{O. Qasem and W. Gao are with the Department of Mechanical and Civil Engineering,
        College of Engineering and Science,
        Florida Institute of Technology, Melbourne, FL 32901, USA
        {\tt\small oqasem2021@my.fit.edu, wgao@fit.edu }}%
\thanks{K. Jebari is with the Department of Aerospace Engineering,
        College of Engineering and Science,
        Florida Institute of Technology, Melbourne, FL 32901, USA {\tt\small kjebari2021@my.fit.edu}}
\thanks{This work was supported in part by the U.S. National Science Foundation under Grant CMMI-2138206. 
The corresponding author is W. Gao.}
}

\maketitle

\begin{abstract}
In this paper, a novel adaptive optimal control strategy is proposed to achieve the cooperative optimal output regulation of continuous-time linear multi-agent systems based on adaptive dynamic programming (ADP). 
The proposed method is different from those in the existing literature of ADP and cooperative output regulation in the sense that the knowledge of the exosystem dynamics is not required in the design of the exostate observers for those agents with no direct access to the exosystem. 
Moreover, an optimal control policy is obtained without the prior knowledge of the modeling information of any agent while achieving the cooperative output regulation. 
Instead, we use the state/input information along the trajectories of the underlying dynamical systems and the estimated exostates to learn the optimal control policy.
Simulation results show the efficacy of the proposed algorithm, where both estimation errors of exosystem matrix and exostates, and the tracking errors converge to zero in an optimal sense, which solves the cooperative optimal output regulation problem.   
\end{abstract}

\section{Introduction}
 In the past bidecade, cooperative control of multi-agent systems (MASs) has gained numerous attentions due to its importance in real-world applications. The cooperative output regulation problem (CORP) is mainly concerned in designing distributed controllers 
 to achieve asymptotic tracking of a class of reference inputs, in addition to rejecting the disturbances in leader-follower MASs. The problem is usually formulated as leader tracking and disturbance rejection, wherein the subsystems (followers) are split into two groups. The first group of followers have direct access to the signal of the exosystem (leader), and the second group consists of those who do not have a direct access to it. The designed controller ensures the stability of the closed-loop system of the whole MASs. 
 \subsection*{Related Works}
Due to its massive impact and effectiveness in engineering applications, the CORP has been widely investigated for both continuous-time linear systems \cite{7403070,cai2017adaptive,deng2020distributed,deng2019distributed,lu2017leader,Huang2012,Su2012Cyber}, and discrete-time linear systems \cite{Yan2016,Liu2018,yan2017cooperative,LIU2019587}. Such applications include connected and autonomous vehicles, cooperative robot reconnaissance, and satellite clustering \cite{ploeg2014graceful, liu2017cooperative, Lu2017TAC,Deng2020}. In addition, 
the output regulation 
has been 
considered for nonlinear systems, see \cite{byrnes1997structurally,Huang2004,Khalil19941587,Serrani2001TAC,xu2017constructive} and references therein. 
 
Generally speaking,  the CORPs are solved using either the feedback-feedforward control strategy, or the internal model principle. Moreover, the cooperative control problem is mainly solved in a distributed way, such that when not all systems in the network have a direct access to the exosystem, those systems reconstruct the exosystem signals through their communication channels with their direct neighbors \cite{Su2013}. For instance, a distributed observer is designed in \cite{Huang2012} to estimate the exogenous signals for the agents with no direct access to exosystems. 
 
 Besides the issues of accessibility and maintaining the asymptotic tracking, having the full knowledge of the dynamics of each agent is a difficult task or impossible practically.
 Moreover, the complexity of modeling a dynamical system increases dramatically as the number of agents and their states increase. 
 In order to address this challenge, the authors in \cite{Baldi2020} have developed an indirect adaptive control approach to solve the CORP with unknown system dynamics.
 However, the designed control policy may be far from being optimal, whether in its transient or steady state since the main objective of 
 to 
 achieve the closed-loop stability with rejecting disturbances. 
These gaps have been filled up in \cite{Gao2017ACC} and \cite{gao2018leader} where a solution to the cooperative optimal output regulation problem (COORP) was proposed, such that a data-driven optimal controller is designed to approximate the feedback and feedforward control gains without the knowledge of MASs' dynamics using the online state/input information collected along the trajectories of each subsystem. Using reinforcement learning (RL) and Bellman's principle of optimality \cite{sutton2018reinforcement}, adaptive dynamic programming (ADP) methods
\cite{gao2021reinforcement,he2019adaptive,JiangBook2017,kamalapurkar2018reinforcement,vamvoudakis2012multi,wei2020continuous,yang2021model,Bian2016,bian2021reinforcement,heydari2018stability,
rizvi2019reinforcement,
zhao2020event,QasemHI,GAO2022110366}
 are developed such that each agent can learn towards the optimal control policy by interacting with its unknown environment. With this learning framework, one can develop an adaptive optimal controller which behaves optimally on a long term without the knowledge of the system matrices. Differential game theory has also been considered with output regulation problems in \cite{odekunle2020reinforcement}. It studies how systems interact with each other and  considers them as players in a game, and provides them with utility functions and learning rules to achieve a collective goal \cite{marden2015game}. 
\subsection*{Main Contributions}

In this work, we propose an innovative adaptive optimal control design algorithm to obtain an approximated optimal feedback-feedforward controller by means of ADP in parallel with the exosystem estimator \cite{Baldi2020} so that each agent can achieve asymptotic tracking while rejecting disturbances without previous knowledge of the agents' and the exosystem's dynamics. First, our proposed algorithm is different from those presented in \cite{liu2017cooperative,Lu2017TAC,Deng2020,Baldi2020}, in the sense that the designed feedback and feedforward gains in our work are optimal and are achieved adaptively. 
Second, comparing to our previous work \cite{Gao2017ACC,gao2018leader}, the prior knowledge of the exosystem dynamics and the frequencies of the exostates is not required. Therefore, this work is the first of its kind to solve the COORPs with adaptive observer using the feedback-feedforward strategy, wherein the designed feedback-feedforward gains are optimal. Third, neither the knowledge of the agents' dynamics nor that of the exosystem is required in the proposed approach with guaranteed convergence analysis of the proposed algorithm provided. Last but not least, the cooperative output regulation is achieved with rigorously stability analysis such that the tracking error along with the exostate estimation errors converge to zero asymptotically. 

\subsection*{Structure}
The rest of this paper is organized as follows. Section \ref{sec: Problem Statement} formulates the problem and covers the preliminaries. The main results of the ADP approach along with the distributed exosystem estimator are presented in Section \ref{sec: COORP}. In Section \ref{sec: Simulation}, simulation results are given to illustrate the efficacy of the proposed algorithm. Last but not least, the conclusion is drawn in Section \ref{sec: Conclusion}.

\subsection*{Notations}
The operator $|\cdot|$ represents the Euclidean norm for vectors and the induced norm for matrices. $\mathbb{Z}_+$ denotes the set of nonnegative integers. The Kronecker product is represented by $\otimes$, and the block diagonal matrix operator is denoted by $\textrm{bdiag}$. $I_n$ denotes the identity matrix of dimension $n$ and $0_{n\times m}$ denotes a $n\times m$ zero matrix. $\text{vec}(A) = [a_1^\textrm{T},a_2^\textrm{T},...,a_m^\textrm{T}]^\textrm{T}$, where $a_i \in \mathbb{R}^n$ is the $i^{\text{th}}$ column of $A \in \mathbb{R}^{n\times m}$. For a symmetric matrix $P=P^\textrm{T} \in \mathbb{R}^{m\times m},$ $\text{vecs}(P)=\left[p_{11},2p_{12},...,2p_{1m},p_{22},2p_{23},...,2p_{m-1,m},p_{mm}\right]^\textrm{T}\in \mathbb{R}^{\frac{1}{2}m(m+1)}$. $P\succ(\succeq)0$ and $P\prec(\preceq)0$ denote the matrix $P$ is positive definite (semidefinite) and negative definite (semidefinite), respectively. For a column vector $v\in \mathbb{R}^m$, $|v|_P = v^\textrm{T}Pv$ , and $\text{vecv}(v)=\textrm{vecs}(v^\textrm{T}v)$. For a matrix $A\in \mathbb{R}^{n\times n}$, $\sigma(A)$ denotes the spectrum of $A$. For any $\lambda \in \sigma(A)$, $\text{Re}(\lambda)$ represents the real part of the eigenvalue $\lambda$.
\section{Problem statement and preliminaries}\label{sec: Problem Statement}
In this section, the problem to be studied is presented with the preliminaries and assumptions considered throughout this paper. To begin with, a class of continuous-time linear MASs are described as follows.
\begin{align}
    \dot v&=Ev,\label{exosys}\\
    \dot{x}_i&=A_ix_i+B_iu_i+D_iv,\label{subsys1}\\
    e_{i}&=C_ix_i+F_iv,~i\in\mathcal{F},\label{subsys2}
\end{align} 
 where $x_i\in\mathbb{R}^{n_i}$ is the state,
 $u_i\in\mathbb{R}^{m_i}$ is the control input, $e_{i}\in\mathbb{R}^{p_i}$ is the tracking error for the $i^{\textrm{th}}$ subsystem, and $v \in \mathbb{R}^{q}$ is the state of the exosystem in \eqref{exosys}. 
The set $\mathcal{F}$ is defined by $\mathcal{F}=\{1,2,\ldots,N\}$. The matrices $A_i\in \mathbb{R}^{n_i\times n_i}$, $B_i\in \mathbb{R}^{n_i{\times m_i}}$, $D_i\in \mathbb{R}^{n_i\times q}$, $C_i\in \mathbb{R}^{p_i\times n_i}$ and ${F_i\in \mathbb{R}^{p_i\times q}}$. The multi-dimensional harmonic oscillator matrix is in the form of
\begin{align}\label{eq: E definition}
    E ={\textrm{bdiag}}\begin{bmatrix}
0 & w_r \\
-w_r & 0 
\end{bmatrix}_{q/2}\in\mathbb{R}^{q\times q},
\end{align}
where $w_r > 0$, $r = 1,...,q/2$ are unknown distinct frequencies of the exosystem. The matrices $A_i$, $B_i$, $D_i$, and $E$ are assumed to be unknown for all $i\in\mathcal{F}$, with the exosystem defined in \eqref{exosys} being marginally stable due to $w_r$ being nonzero and distinct.

The digraph $\mathcal{G} = \{\mathcal{V},\mathcal{E}\}$, where $\mathcal{V}= \{0,1,…,N\}$ is the set of nodes with $0$ denoting the leader modeled by \eqref{exosys}, and $\mathcal{F}$ represents the set of followers modeled by \eqref{subsys1} and \eqref{subsys2}. $\mathcal{E}$ represents the edge set $\mathcal{E}\subset \mathcal{V} \times \mathcal{V}$ where an edge from node $i$ to node $j$ is denoted by $(i,j)$ and $\mathcal{N}_i$ denotes the subset of $\mathcal{V}$ which consists of all the neighbors of the $i^{\textrm{th}}$ node.
The adjacency matrix of the digraph $\mathcal{G}$ is denoted by $\mathcal{A}$, such that $\mathcal{A} =[ a_{ij} ]$ satisfies $a_{ii}=0$ and $a_{ij}>0$ when $(j,i) \in \mathcal{E}$.
The Laplacian matrix of digraph $\mathcal{G}$ is denoted by $\mathcal{L}=[l_{ij} ]$, where $l_{ii}=\sum_{j=1}^{N} a_{ij}$  and $l_{ij}=-a_{ij}$ if $i \neq j$. The target matrix $\mathcal{M}=[m_{ij}]$ represents the communication links between the leader and the followers, wherein $m_{ii}=1$ if $i\in \mathcal{T}$ and $m_{ii}=0$ otherwise, where $\mathcal{T}$ is a set of nodes with direct contact with the leader.

Our main goal is to design a data-driven distributed optimal control policy for the MASs described by \eqref{exosys}-\eqref{subsys2} to solve the COORP. Inspired by our previous work \cite{Gao2017ACC}, the proposed method is built upon an ADP approach. In particular, policy iteration (PI) is used, whose rate of convergence is quadratic \cite{Kleiman1968}. Different from \cite{Gao2017ACC}, we drop the assumption of the $a$ $priori$ knowledge of the exosystem matrix. 
Furthermore, the exostates are not accessible by all the followers. 
Some standard assumptions are taken into consideration while solving the CORP, which are listed as follows.
\begin{assumption}\label{assumption: controllability}
The pairs ($A_i$,$B_i$) and ($C_i$,$A_i$) are stabilizable and observable for all $i\in\mathcal{F}$.
\end{assumption}
\begin{assumption}\label{graphasm}
The leader interacts with at least one follower, i.e., $\mathcal{T}$ is nonempty, and the graph $\mathcal{G}$ is undirected and connected if eliminating the leader node and all edges connected to the leader.
\end{assumption}
\begin{assumption}\label{assumption: rank}
rank$\left(\begin{bmatrix}
A_i-\lambda I & B_i\\C_i&0
\end{bmatrix}\right) = n_i +p_i,\,\forall\lambda \in \sigma (E)$, and $i\in\mathcal{F}$.
\end{assumption}
\begin{remark} 
Unlike \cite{Baldi2020}, where the matrix $F$ has to be in the form $F=\left[{0\;1}\ldots\;0\;1\right]$, such that the block $\left[0\;1\right]$ is repeated $q/2$ times, 
the ADP-based approach in our work is more 
flexible and has no restrictions on the structure of the matrix $F$.
\end{remark}
\begin{remark}
Assumption \ref{graphasm} guarantees that at least one follower has access to the exogenous signals, and that all subsystems have access to at least one of their neighbors' exosystem estimation $\eta_j$, $\forall j \in \mathcal{N}_i$ so that $\epsilon_i \neq 0$ $\forall i \in \mathcal{V}$.
\end{remark}

If the exostate is accessible by all the followers, the CORP can be solved by a decentralized control policy in the form of  
\begin{align}
    u_i&=-K_ix_i+L_iv,~i\in\mathcal{F},\label{uctrl}
\end{align}
where $K_i$ and $L_{i}$ are the feedback and forward gain matrices, respectively.
\begin{definition}
For any $i\in\mathcal{F}$, a control feedback gain matrix $K_i \in \mathbb{R}^{m_i\times n_i}$ is called stabilizing for the $i^{\textrm{th}}$ subsystem if and only if $A_i-B_iK_i$ is Hurwitz.
\end{definition}

Note that the decentralized controller (\ref{uctrl}) is not applicable according to the Assumption \ref{graphasm}. If the exosystem matrix $E$ is known to all followers, then one may design a distributed observer \cite{Huang2012, Gao2017ACC} to estimate the exostate $v$.
However, if the exact knowledge of the exosystem matrix is not available, it is impossible to use the distributed observer proposed in \cite{Huang2012, Gao2017ACC} to estimate the exostates signals. This barrier has been removed by proposing a distributed adaptive estimator in \cite{Baldi2020}, in which the estimation does not require a prior knowledge of the exosystem matrix. 
We will take the advantage of this result to develop a new ADP algorithm to solve the COORP with adaptive observer.
\section{Solving COORP with Adaptive Observer}\label{sec: COORP}
In this section, an ADP-based approach with distributed adaptive observer is proposed with stability and convergence analysis provided. To begin with, we consider observing the unknown exosystem states. The local observation error for the agent $i$ is defined as follows.
\begin{align}
   \epsilon_i=\sum_{j=1}^{N}a_{ij}(\eta_i-\eta_j)+m_{ii}(\eta_i-v),\label{epsilon}
\end{align}
such that $\eta_i$ and $\epsilon_i $ are vectors in the form of
\begin{align}
    \eta_i&=\begin{bmatrix}
\eta_{i,1} & \ldots &\eta_{i,q} 
\end{bmatrix}^\textrm{T},\;\;
    \epsilon_i=\begin{bmatrix}
\epsilon_{i,1} & \ldots &\epsilon_{i,q} 
\end{bmatrix}^\textrm{T}.
\end{align}

By having $\epsilon_i \rightarrow 0\; \forall i \in \mathcal{V}$, it is then achievable that $\eta_i \rightarrow v, \forall i \in \mathcal{V}$. This enables us to reconstruct the exostate $v$ for non-target nodes that do not have access to it.
The distributed adaptive observer is as follows.
\begin{align}
    \dot \eta_i&=\hat{E_i}\eta_i+(\mathbb{A}_m-\hat{E}_i)\epsilon_i,\label{etadot}
\end{align}
with $ \mathbb{A}_m \in \mathbb{R}^{q \times q}$ a Hurwitz diagonal matrix defined as
\begin{align}
    \mathbb{A}_m&=-{ \textrm{bdiag}}(a_r.I_2)_{q/2},\;a_r>0,\;r=1,\ldots,q/2.\label{A_m}
\end{align}
The estimation of the exosystem matrix $E$ for the $i^{\text{th}}$ follower is
\begin{align}
    \hat{E_i}&=\textrm{bdiag}\begin{bmatrix}
    0 & (\hat{w}_r)_i \\
    -(\hat{w}_r)_i & 0 
   \end{bmatrix}_{q/2},\label{Ehat}
  \end{align}
where $\left(\hat{w}_r\right)_i$ is the estimate of the frequencies of the leader system for the $i^{\text{th}}$ subsytem described by the following dynamical equation. 
\begin{align}
   \left(\dot{\hat{w}}_r\right)_i&=\kappa_r\left(\eta_{i,(2r-1)}\epsilon_{i,(2r)}-\eta_{i,(2r)}\epsilon_{i,(2r-1)}\right),\label{What}
  \end{align}
with $\kappa_r>0$ being a constant design gain.

\begin{lemma}[\hspace{-0.2pt}\cite{Baldi2020}]\label{lemma: observer} Under Assumption \ref{graphasm}, by considering \eqref{etadot}-\eqref{What} the adaptation law \eqref{What} guarantee that  $\eta_i \rightarrow v$ and $\hat{E}_i \rightarrow E$ as $t \rightarrow \infty,\; \forall i \in \mathcal{V}$.  
\end{lemma} 

With the estimated exosystem matrix and exostate, we will introduce the ADP strategy to compute the optimal feedback-feedforward control policy to solve the COORP.

 The COORP studied considers both transient and steady state responses of each subsystem. The formulation follows the traditional linear optimal control output regulation problem \cite{Gao2017ACC} that the optimal distributed output regulation problem needs to solve the following two problems besides the CORP. 

  \begin{problem}\label{Problem 1}
  \begin{align}
      \min_{(X_i,U_i)}\textrm{ Tr}&(X_i^\textrm{T}\bar{Q}_iX_i+U_{i}^\textrm{T}\bar{R}_iU_{i}), \label{min_Tr}\\
   \text{subject to    }  X_iE&=A_iX_i+B_iU_i+D_i,\label{regeq1}\\
     0&=C_iX_i+F_i,\label{regeq2} 
  \end{align}
  where $\bar Q_i=\left(\bar Q_i\right)^\textrm{T}\succ0$ and $\bar R_i=\left(\bar R_i\right)^\textrm{T}\succ0$, for all 
  $i\in\mathcal{F}.$
   \end{problem}
   
  Based on Assumption \ref{assumption: rank}, the solvability of the regulator equations defined by \eqref{regeq1}-\eqref{regeq2} is guaranteed and the pairs $(X_{i},U_{i})$ exist for any matrices $D_{i}$ and $F_{i}$, $\forall$ $i\in\mathcal{F}$; see \cite{huangjiebook}. 
  Additionally, the solution to Problem \ref{Problem 1}, i.e., $(X_{i}^\star,U_{i}^\star)$ is unique.  
  \begin{problem}
  \begin{align}
      \min_{\bar{u}_i}  \int_{0}^{\infty} &\left(|e_{i}|_{Q_i}+|\bar{u}{_i}|_{R_i}\right) \,\textrm{d}t, \\\label{eq: error system 1}
      \textrm{subject to } \dot{\bar{x}}_{i}&=A_{i}\bar{x}_i+B_{i}\bar{u}_{i},\\\label{eq: error system 2}
      e_{i}&=C_{i}\bar{x}_{i},
  \end{align}
  where $Q_i=\left(Q_i\right)^\textrm{T} \succeq 0, R_i =\left(R_i\right)^\textrm{T} \succ0,$ with $\left(A_i,\sqrt{Q_i}C_i\right)$ being observable for all $i\in\mathcal{F}$. The equations \eqref{eq: error system 1}-\eqref{eq: error system 2} form the error system with $\bar{x}_{i}:=x_{i}-X_iv$ and $\bar{u}_{i}:=u_{i}-U_iv$.
  \end{problem}
  
Note that if the followers' dynamics in (\ref{subsys1}) are known, one can develop a distributed optimal controller in the following form.
\begin{align}\label{eq: ui*}
    u_i^\star(K_{i}^\star,L_{i}^\star)=-K_{i}^\star x_i+L_{i}^\star\eta_i,
\end{align} where $K_{i}^\star=R_{i}^{-1}B_{i}^\textrm{T}P_{i}^\star$, and $P_{i}^\star $ is the unique solution of the following albegraic Ricatti equation (ARE)
\begin{align}
    A_i^\textrm{T}P_i^\star+P_i^\star A_i+C_i^\textrm{T}Q_iC_i-P_i^\star B_iR_i^{-1}B_i^\textrm{T}P_i^\star=0.\label{eq: ARE}
\end{align}

The solutions to the regulator equations \eqref{regeq1}-\eqref{regeq2}, i.e., $(X_i,U_i)$, form the optimal feedforwad gain matrix such that \begin{align}\label{eq: Li*}
    L_{i}^\star&=U_{i}+K_{i}^\star X_{i},~\forall i\in\mathcal{F}.
\end{align}

It is remarkable that equation \eqref{eq: ARE} is nonlinear in $P_{i}^\star $. Therefore, in this paper we consider an iterative method to solve $P_{i}^\star $, i,e., ADP. In particular, we use PI since the rate of convergence of PI is quadratic, since it is a Newton-Raphson based method. In PI, the iterative process to find the optimal control policy is done by alternating two stages, i.e., policy evaluation, and policy improvement. The following lemma shows the convergence of \eqref{eq: ARE} in the sense of the PI method.
\begin{lemma}[\hspace{-0.2pt}\cite{Kleiman1968}]
Let $K_{i,0} \in \mathbb{R}^{m_i \times n_i}$ be a stabilizing feedback gain matrix $\forall~  i\in\mathcal{F}$, the matrix $P_{i,k}=(P_{i,k})^\textrm{T}\succ0$ be the solution of the following equation  
\begin{align}\label{eq: Policy evaluation}
   P_{i,k}(A_i-B_iK_{i,k-1})+(A_i-&B_iK_{i,k-1})^\textrm{T} P_{i,k} +C_i^{\textrm{T}}Q_iC_i\nonumber\\&+K_{i,k-1}^\textrm{T} R_i K_{i,k-1}=0,
\end{align} 
 and the control gain matrix $K_{i,k}$, with $k=1,2,\cdots,$ are defined recursively by 
\begin{align}\label{eq: Policy improvement}
    K_{i,k}=R_{i}^{-1} B_i^\textrm{T} P_{i,k-1}.
\end{align} 
Then the following properties hold for any $k\in \mathbb{Z}_{+}$, $ i\in\mathcal{F}$ 
\begin{enumerate}
    \item The matrix $A_i-B_iK_{i,k}$ is Hurwitz.
    \item $P_i^\star \preceq P_{i,k} \preceq P_{i,k-1}$.
    \item $\underset{{k\rightarrow \infty}}\lim K_{i,k} = K_i^\star ,\; \underset{{k\rightarrow \infty}}\lim P_{i,k}=P_i^\star $.
\end{enumerate}
\end{lemma}

Since the unavailability of the system and exosystem dynamics is considered in our approach, the states, inputs and exostates information collected along the trajectories of the underlying dynamical systems are used to learn the optimal feedback and feedforward gain matrices $K_{i}^\star $ and $L_{i}^\star $ for all $i\in\mathcal{F}$. By solving $U_{i}^\star $ and $X_{i}^\star $ from \eqref{regeq1}-\eqref{regeq2}, one is able to solve for $L_{i}^\star $ from \eqref{eq: Li*}.

Now we are ready to introduce the variables $\bar{x}_{ij}=x_i-X_{ij}v,$ $ j = 0,1,2,...,h_i+1$ with $h_i=(n_i+p_i)q$ being the dimension of the null space of $(I_q \otimes C_i)$, where the following is met: $X_{i0}=0_{n_i \times q},$ $X_{i1} \in \mathbb{R}^{n_i \times q}$ such that $C_iX_{i1}+F_i=0$, and $X_{ij} \in \mathbb{R}^{n_i \times q}$, $\forall j \in{2,3,...,h_i+1}$, such that all $\text{vec}(X_{ij})$ form a basis for $\text{ker}(I_q \otimes C_i)$, where $\text{ker}(\cdot)$ denotes the null space. 

The definition of $\bar{x}_{ij}$ enbales us to solve the Sylvester map of trail matrices $X_{ij}$ which is in itself crucial for solving the regulator equations to finally approximate $L^\star $ and $K^\star $ without previous knowledge of the systems' matrices ($A_i$, $B_i$, and $D_i$). By taking the time derivative along the trajectories of $\bar{x}_{ij}$ we have 
\begin{align}
   \dot{\bar{x}}_{ij}&=\dot{x}_i-X_{ij}\dot{v}\nonumber\\
   &=A_ix_i+B_iu_i+D_iv-X_{ij}Ev\nonumber\\
   &=A_{i,k}\bar{x}_{ij}+B_i(K_{i,k}\bar{x}_{ij}+u_i)+(D_i-S_i(X_{ij}))v,\label{xbarij_dot}
  \end{align}
  where $A_{i.k}=A_i-B_iK_{i,k}$ and $S_i(X)=XE-A_iX$ is a Sylvester map, $S_i: \mathbb{R}^{n_i \times q} \rightarrow \mathbb{R}^{n_i \times q}$. By taking the integration over the time inerval $[t,t+\delta t]$, $\delta t>0$, we obtain
  \begin{align}\label{eq: integration}
      &|\bar{x}_{ij}(t+\delta t)|_{P_{i,k}}-|\bar{x}_{ij}(t)|_{P_{i,k}}\nonumber\\&=\hspace{-3pt}\int_{t}^{t+\delta t}\hspace{-4.5pt}\left(|\bar{x}_{ij}|_{(A_{i,k}^\textrm{T}P_{i,k}+P_{i,k}A_{i,k})}+2(u_i+K_{i,k}\bar{x}_{ij})^\textrm{T}B_i^\textrm{T}P_{i,k}\bar{x}_{ij} \right.\nonumber\\&\left.+2v^\textrm{T}(D_i-\mathcal{S}_i(X_{ij}))^\textrm{T}P_{i,k}\bar{x}_{ij}\right )\textrm{d}\tau \nonumber\\&=\hspace{-3pt}\int_{t}^{t+\delta t}\hspace{-4.5pt}\left(-|\bar{x}_{ij}|_{(Q_i+K_{i,k}^\textrm{T}R_iK_{i,k})}+2(u_i+K_{i,k}\bar{x}_{ij})^\textrm{T}R_iK_{i,k+1}\bar{x}_{ij} \right.\nonumber\\&\left.+2v^\textrm{T}(D_i-\mathcal{S}_i(X_{ij}))^\textrm{T}P_{i,k}\bar{x}_{ij}\right)\textrm{d}\tau. 
  \end{align}
  Using Kronecker product properties, \eqref{eq: integration} can be written in a compact form as follows, wherein the approximated values of $P_{i,k}$ and $K_{i,k+1}$ can be solved in the sense of least square errors 

\begin{align}
   \Psi_{ijk} \left [\begin{array}{cc}
        \text{vecs}(P_{i,k})  \\
        \text{vec}(K_{i,k+1})\\
        \text{vec}((D_i-S_i(X_{ij})^\textrm{T}P_{i,k})
   \end{array} \right]=\Phi_{ijk},  \label{adp_eq}
  \end{align}
  where
 \begin{align}
 \Psi_{ijk}&=[\delta_{\bar{x}_{ij}\bar{x}_{ij}},-2\Gamma_{\bar{x}_{ij}\bar{x}_{ij}}(I_{ni}\otimes K_{i,k}^\textrm{T}R_i)-2\Gamma_{\bar{x}_{ij}u_i}(I_{ni} \otimes R_i),\nonumber\\&\;\;\;\;\;-2\Gamma_{\bar{x}_{ij}v}], \nonumber\\
\Phi_{ijk}&=-\Gamma_{\bar{x}_{ij}\bar{x}_{ij}}\text{vec}\left(Q_i+K_{i,k}^\textrm{T}R_iK_{i,k}\right),\nonumber\\
\delta_a&=\left[\text{vecv}(a(t_1))-\text{vecv}(a(t_0)),\ldots,\right.\nonumber\\&\;\;\;\;\;\;\;\;\;\;\;\;\;\;\;\;\;\;\;\;\;\;\;\;\;\;\;\;\;\;\;\;\;\;     \;\left.\text{vecv}(a(t_s))-\text{vecv}(a(t_{s-1}))\right]^\textrm{T},\nonumber\\
\Gamma_{a,b}&=\left[\int_{t_0}^{t_1} a\otimes b  \,\textrm{d}\tau,\int_{t_1}^{t_2} a\otimes b  \,\textrm{d}\tau,\ldots,\int_{t_{s-1}}^{t_{s}} a\otimes b  \,\textrm{d}\tau\right]^\textrm{T}.\nonumber
 \end{align}
 
The time sequence $\{t_l\}_{l=0}^{s}$ is a strictly increasing sequence. The uniqueness of the solution to equation \eqref{adp_eq} is guaranteed when the following rank condition is met:
  {\begin{align}
  \hspace{-2.94mm} \textrm{rank}\left(\left[\Gamma_{\bar{x}_{ij}\bar{x}_{ij}},\Gamma_{\bar{x}_{ij}u_i},\Gamma_{\bar{x}_{ij}v}\right]\right)=\frac{n_i(n_i+1)}{2}+(m_i+q_i)n_i.\hspace{-1.5mm}\label{rank_cndt}
  \end{align}}
  \begin{remark}
  In order to satisfy the condition in \eqref{rank_cndt}, exploration noise is added to the applied input during the learning phase. The exploration noise is usually random noise, random sinusoidal signals, or summation of sinusoidal signals with different frequencies. 
  \end{remark}

The general solution to the regulator equations \eqref{regeq1}-\eqref{regeq2} is obtained from the following.
\begin{align}
    X_i&=X_{i1}+\sum_{j=2}^{h_i+1} \alpha_{ij}X_{ij},\;\;\;\;\;\;\alpha_{ij} \in \mathbb{R}, \nonumber\\
    S_i(X_i)&=S_i(X_{i1})+\sum_{j=2}^{h_i+1}\alpha_{ij}S_i(X_{ij})=B_iU_i+D_i.\label{solutiosreg}
\end{align}
In matrix form, the equation in \eqref{solutiosreg} can be written as
\begin{align}
    \mathcal{A}_i\mathcal{X}_i=b_i,\label{solutionregequ}
\end{align}
where
\begin{align}
\mathcal{A}_i&=\begin{bmatrix}\mathcal{A}_{i1}&\mathcal{A}_{i2} \end{bmatrix},\nonumber\\
    \mathcal{A}_{i1}&=\begin{bmatrix} \text{vec}(S_i(X_{i2}))& \ldots & \text{vec}(S_i(X_{i,h+1}))\\\text{vec}(X_{i2})&\ldots &\text{vec}(X_{i,h+1})\end{bmatrix},\nonumber\\
    \mathcal{A}_{i2}&=\begin{bmatrix}0&-I_q\otimes (P_{ik}^{-1}K_{i,k+1}R_i)\\-I_{n_iq}&0 \end{bmatrix},\nonumber\\
    \mathcal{X}_i&=\begin{bmatrix} \alpha_{i2},\ldots,\alpha_{i,h+1},\text{vec}(X_i)^\textrm{T},\text{vec}(U_i)^\textrm{T} \end{bmatrix}^\textrm{T},\nonumber\\
    b_i&=\begin{bmatrix}\text{vec}(-S_i(X_{i1})+D_i)\\-\text{vec}(X_{i1})  \end{bmatrix}.\nonumber
\end{align}

In Theorem \ref{stability theorem}, we show that although the estimations of the exostates are used instead of their actual values, the cooperative output regulation can be achievable. Furthermore, the ADP algorithm to solve the COORP with the adaptive observer is shown in Algorithm \ref{algorithm: ADP} with proof of convergence shown in Theorem \ref{theorem}.
\begin{theorem}\label{stability theorem}
Given Assumptions \ref{assumption: controllability}-\ref{assumption: rank}, and under the system described by \eqref{exosys}-\eqref{What}, if ${\mathbf{\bar A}}_i=A_{i}-B_{i}K_{i}$ is Hurwitz ${\forall i\in\mathcal{F}}$, then the closed-loop state-feedback controller $u_{i}=u_{i}^{\star}(K_{i},\hat L_{i})$ achieves the cooperative output regulation,
where $\hat{L}_i=K_i\hat X_{i}+\hat{U}_{i}, $ and the pairs $(\hat{X}_i,\hat{U}_i)$ are the solutions of the following regulator equations. \begin{align}
    \hat{X}_{i}\hat{E}_{i}&=A_i\hat{X}_i+B_i \hat{U}_{i}+D_{i},\nonumber\\
    0&=C_{i}\hat X_{i}+F_i,\forall{i\in\mathcal{F}}.\nonumber
\end{align} 
\end{theorem}
\begin{proof}
Let $\tilde{x}_i=x_{i}-\hat{X}_i\eta_i$, $\tilde{u}_i=u_{i}-\hat{U}_i\eta_i$ and $\tilde{\eta}_i=v-\eta_i$. By taking the time derivative of $\tilde{x}_i$ the following equation is obtained. \begin{align}
    \dot{\tilde{x}}_i=& \dot x_i-\hat{X}_i\dot\eta_i\nonumber\\
    =& A_ix_i+B_iu_i+D_{i}v_{i}-\hat{X}_{i}\left(\hat{E}_{i}\eta_i+\left(\mathbb{A}_m -\hat{E}\right)\epsilon_i\right)\nonumber\\
    =& {\bar{\mathbf{A}}}_i\tilde{x}_i+B_i(\tilde{u}_i+K_i\tilde{x}_i)+D_i\tilde{\eta}_i-\hat{X}_i(\mathbb{A}_m-\hat{E}_i)\epsilon_i\nonumber\\\label{eq: tilde x}
    =& {\bar{\mathbf{A}}}_i\tilde{x}_i+B_i(\tilde{u}_i+K_i\tilde{x}_i-\hat{L}_i\tilde{\eta}_i)+(D_i+B_i\hat{L}_i)\tilde{\eta}_i\nonumber\\
    &-\hat{X}_i(\mathbb{A}_{m}-\hat{E}_i)\epsilon_{i}.
\end{align}
By \eqref{eq: ui*}, we have $\tilde{u}_i=-K_i\tilde{x}_i+\hat{L}_i\tilde\eta_i$. In addition, since $v$ is bounded, so is $\tilde{\eta}_i\to 0$. Given $\bar{\mathbf{A}}_i$ is Hurwitz, we have \eqref{eq: tilde x} is input to state stable with $-\hat{X}_i(\mathbb{A}_m-\hat{E}_i)\epsilon_i$ and $\tilde\eta_i$ as inputs. In other words, there exist a function $\beta$ of class $\mathcal{K}\mathcal{L}$ and a function $\gamma$ of class $\mathcal{K}$ such that \begin{align}
    |\tilde{x}_i(t)|\leq \beta(|\tilde{x}_i(0)|,t)+\gamma\left(\underset{0\leq \tau \leq t}\sup\{|\tilde\eta_i(\tau)|, |\mathcal{Q}_i(\tau)|\}\right),
\end{align}
where $\mathcal{Q}_i=\hat{X}_i(\mathbb{A}_{m}-\hat{E}_i)\epsilon_{i}$.
Therefore, with the existence of the estimate of $v$, $\tilde{x}$ remains bounded. In addition, on the basis of \cite{Baldi2020}, we have ${\underset{t\to \infty}\lim}\epsilon_i = 0$ and ${\underset{t\to \infty}\lim}\tilde\eta_i = 0$  ${\forall i\in\mathcal{F}}$. Hence, it is concluded that $
\underset{t\to\infty}\lim\tilde{x}_{i}=0$ and $\underset{t\to\infty}\lim\tilde{e}_{i}=0$. The proof is thus completed. 
\end{proof}


\begin{algorithm}[tb!]
\begin{algorithmic}[1]
\small
\State $i \leftarrow 1$
 \Repeat
\State Compute the estimation $\hat{E}_i$ from \eqref{epsilon}-\eqref{What}. ${E_{i}\gets \hat{E}_i}$.
\State Choose $K_{i,0}$ such that $A_i-B_iK_{i,0}$ is a Hurwitz matrix,

\hspace{-0.2cm}and a small constant $\varepsilon_i > 0$.
\State Apply an essentially bounded control input 

\hspace{-0.25cm}${u_{i,0}=-K_{i,0}x_i+\zeta_i}$ where $\zeta_i$ is the exploration noise, 

\hspace{-0.25cm}over a time interval $[t_0,t_s]$. {$j \leftarrow 0$.}
\Repeat
\State Compute $\Gamma_{\bar{x}_{ij}\bar{x}_{ij}}, \Gamma_{\bar{x}_{ij}u_i}, \Gamma_{\bar{x}_{ij}v},  \delta_{\bar{x}_{ij}\bar{x}_{ij}}$ such that \eqref{rank_cndt} \text{ \text{  \text{ \text{ \text{}     }        }      }            } holds. {$j \leftarrow j+1$.}
\Until $j=h_i+2$
\State $j \leftarrow 1$, $k \leftarrow 0$
\Repeat
\State Solve $P_{i,k}$ and $K_{i,k+1}$ from \eqref{adp_eq}. {$k \leftarrow k+1$.}
\Until $|P_{i,k}-P_{i,k-1}|<\varepsilon_i$.
\State $k  \gets k^\star,j \gets 1$
\Repeat
\State Solve ${S}_i(X_{ij})$ from \eqref{adp_eq}. {$j \gets j+1$.}
\Until $j=h_i+2$
\State Solve $(X_i^\star ,U_i^\star )$ from Problem 1. 
 \State $L_{i,k^\star} \gets U_i^{\star} +K_{i,k^\star} X_i^\star $
\State Compute the approximated optimal control policy \text{ \text{             \text{        \text{\text{                          }}}}}${u_{i,k^\star} =u^\star (K_{i,{k}^\star} ,L_{i,{k}^\star} )}$ from \eqref{eq: ui*}. {$i\gets i+1$.}
\Until $i=N+1$
\end{algorithmic}
\caption{Data-driven COORP with Adaptive Observer}
\label{algorithm: ADP}
\end{algorithm}
\begin{theorem}\label{theorem}
If the rank condition in \eqref{rank_cndt} is satisfied, then for any small constant $c>0$ there exist constants $\kappa_r>0$, $r=1,2,\ldots,q/2$, and $k^\star \in \mathbb{Z}_+$ such that the sequences $\{P_{i,k}\}_{k=0}^{\infty}$ and $\{K_{i,k}\}_{k=1}^{\infty}$ learned from Algorithm \ref{algorithm: ADP} satisfy the inequalities
$|P_{i}^{(k^\star )}-P_{i}^\star |<c$ and $|K_{i}^{(k^\star )}-K_{i}^\star |<c$, respectively.
\end{theorem}
\begin{proof}
The condition \eqref{rank_cndt} ensures that \eqref{adp_eq} has a unique solution. As the convergence of steps 6-12 has been shown in \cite{Gao2017ACC}, we can always find a small constant $c_1>0$ such that the pairs $(\bar{P}^{k^\star }_i,\bar{K}^{k^\star }_i)$ obtained from steps 6-12 are close enough to $(P_{i}^\star, K_{i}^\star)$ solved from \eqref{eq: Policy evaluation}-\eqref{eq: Policy improvement} satisfying the inequalities $|\bar{P}_{i}^{(k^\star )}-P_{i}^\star |<c_1$ and $|\bar{K}_{i}^{(k^\star )}-K_{i}^\star |<c_1$, for all ${i\in\mathcal{F}}$. Moreover, from Lemma \ref{lemma: observer} and Theorem \ref{stability theorem}, it is always guaranteed that there exists a constant $c_2>0$ such that $|\bar{P}_{i,k}-P_{i,k}|<c_2$ and $|\bar{K}_{i,k}-K_{i}|<c_2$, for every $k\in\mathbb{Z}_{+}$ and ${i\in\mathcal{F}}$, since $(\hat{w}_r)_{i}$ defined in \eqref{What} is uniformly bounded \cite[Theorem 1]{Baldi2020}. 
Using the triangular inequality, we can find an iteration index $k^\star $ and a small constant $c>0$ such that the inequalities $|P_{i}^{(k^\star )}-P_{i}^\star |<c$ and $|K_{i}^{(k^\star )}-K_{i}^\star |<c$ are satisfied. The proof is thus completed.
\end{proof}

\begin{remark}
It is worth mentioning that steps 10-16 in Algorithm \ref{algorithm: ADP} are PI-based. A value iteration (VI) based method, similar to the one developed in our previous work in \cite{gao2021TNN}, can be used to replace these steps. In the VI-based framework, an initial stabilizing policy is not required. However, its convergence rate is slower than the quadratic convergence rate of PI used in Algorithm \ref{algorithm: ADP}. 
\end{remark}
\begin{remark}
The proposed Algorithm \ref{algorithm: ADP} is an off-policy learning method. Each follower learns its own optimal policy independently, which makes it more practical, especially for large scale systems. 
\end{remark}
\section{Simulation and Results}\label{sec: Simulation}
In this section, we illustrate the efficacy of the proposed Algorithm \ref{algorithm: ADP} in an example, in which the system consists of four followers and a leader as depicted in Fig. \ref{fig:system}. The exosystem ($\#0$) is a harmonic oscillator described by the matrix $E$ and the followers ($\#1-4$) are described by the matrices $A_i,B_i,C_i,D_i,$ and $F_i$. 
\begin{figure}[tb!]
    \centering
    \includegraphics[width=0.71\linewidth]{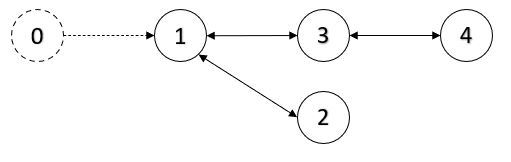}
    \vspace{-1.5mm}
    \caption{The communication topology of the overall system.}
    \label{fig:system}
     \vspace{-3mm}
\end{figure}
In this example we assume that there is no prior knowledge of the dynamics of the system ($A_i, B_i,$ and $D_i$), or the exosystem dynamics ($E$). The system matrices for the system described in \eqref{exosys}-\eqref{subsys2} are shown below for simulation purposes.
\begin{align*}
    A_i&=\begin{bmatrix}
1&1+i&0\\0&2&-0.5i\\1&0&1+i
\end{bmatrix},\;B_i=\begin{bmatrix}
0\\1\\i
\end{bmatrix},\;C_i=\begin{bmatrix}
1/i&0&0
\end{bmatrix},\\D_i&=\begin{bmatrix}
1&0&-1&0\\0&0&1.5i&0\\0&1&0&-0.5i
\end{bmatrix},\;
F_i=\begin{bmatrix}
-0.75i&0&1&0
\end{bmatrix},
\end{align*}
and $E=\textrm{bdiag}\left(\begin{bmatrix}
0&1\\
-1&0\end{bmatrix},
\begin{bmatrix}
0&0.75\\
-0.75&0
\end{bmatrix}\right).$

The weighting matrices of the cost function are $Q=I_{3}$ and $R=1$, the initial values are $v(0)=\begin{bmatrix}0&1&0&0.5\end{bmatrix}^\textrm{T}$, $(\hat{w}_r)_i(0)=0$, and the rest of the parameters are $\varepsilon_i=10^{-4}$,  ${\kappa_r=\begin{bmatrix}40&40\end{bmatrix}}$ and $ a_r=\begin{bmatrix}15&15\end{bmatrix}$, $\forall i=1,2,3,4$.
During the time interval ${0 \leq t \leq 8 s}$, an essentially bounded exploration noise $\zeta_i$ is added to the applied initial control policy. Using Algorithm \ref{algorithm: ADP}, first $\hat{E}_i$ is estimated, then the approximations of the optimal feedback and feedforward control gain matrices $K_{i}^{\star}$ and $L_{i}^{\star}$ are calculated, respectively. Fig. \ref{fig:convergence} depicts that $P_{i,k}$ obtained by Algorithm \ref{algorithm: ADP} converge to their optimal values $P_{i}^\star $ obtained by solving directly from \eqref{eq: ARE}, and the convergence is achieved in less than or equal to 19 iterations. The optimal solution to the regulator equations obtained is 
used to calculate the feedforward gains, which are shown with their corresponding actual ones for the sake of comparison.
\begin{align*}
    L_{1}^{(14)}&=\begin{bmatrix}2.8801&-11.9485&16.4917&12.4644\end{bmatrix},\\
     L_{1}^{\star}&=\begin{bmatrix}2.8801&-11.9484&16.4918&12.4641\end{bmatrix},\\
    L_{2}^{(16)}&=\begin{bmatrix}1.0720&-6.2090&15.1043&7.4341\end{bmatrix},\\
    L_{2}^{\star}&=\begin{bmatrix}1.0721&-6.2089&15.1043&7.4340\end{bmatrix},\\
    L_{3}^{(17)}&=\begin{bmatrix}-3.1127&-7.3517&13.5064&5.2960\end{bmatrix},\\
    L_{3}^{\star}&=\begin{bmatrix}-3.1117&-7.3508&13.5063&5.2923\end{bmatrix},\\
    L_{4}^{(19)}&=\begin{bmatrix}-8.5758&-9.4777&13.3007&4.4879\end{bmatrix},\\
    L_{4}^{\star}&=\begin{bmatrix}-8.5725&-9.4729&13.3089&4.4654\end{bmatrix}.
\end{align*}
From Fig. \ref{fig:convergence} and the above-mentioned feedforward gain matrices, it can be noticed that the approximated control policy converges to the optimal policy, while neither the system dynamics nor that of the exosystem are known. Moreover, Fig. 
\ref{fig:v_and_trackingerror} demonstrates the convergence of the tracking error and the estimation error.
Finally, one can observe from Fig. \ref{fig:output} that all the followers can achieve asymptotic tracking while rejecting nonvanishing disturbance.
\section{Conclusion}\label{sec: Conclusion}
This paper studies the cooperative output regulation problem of a class of continuous-time linear multi-agent systems with unknown system dynamics. A distributed control policy is derived by first estimating the exosystem dynamics for each follower, then adaptive dynamic programming (ADP) is used to approximate the optimal solution to the regulator equations. The effectiveness of the proposed algorithm and the ability to achieve asymptotic tracking while rejecting nonvanishing disturbances are demonstrated by both theoretical analysis and performed simulation. Future work includes extending this work to a class of nonlinear systems with robust analysis subject to external disturbances.

\begin{figure}[tb!]
    \centering
    \includegraphics[width=0.85\linewidth,trim={1.10cm 1mm 1.10cm 4.50mm},clip]{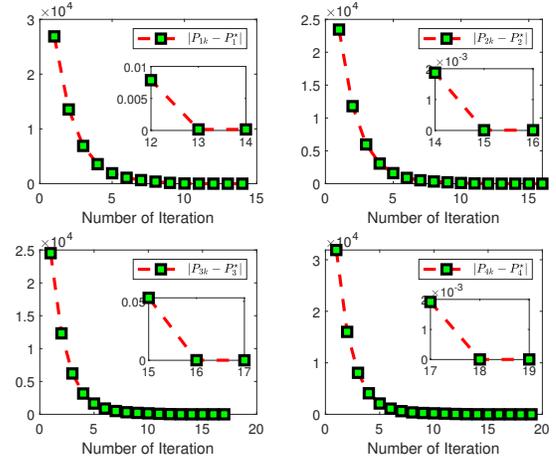}
    \caption{$|P_{i,k}-P_{i}^\star|$ for all $i=1,2,3,4$ under Algorithm \ref{algorithm: ADP}.}
    \label{fig:convergence}
\end{figure}
\begin{figure}[tb!]
    \centering
    \includegraphics[width=0.85\linewidth,trim={0.85cm 1mm 1.10cm 4.50mm},clip]{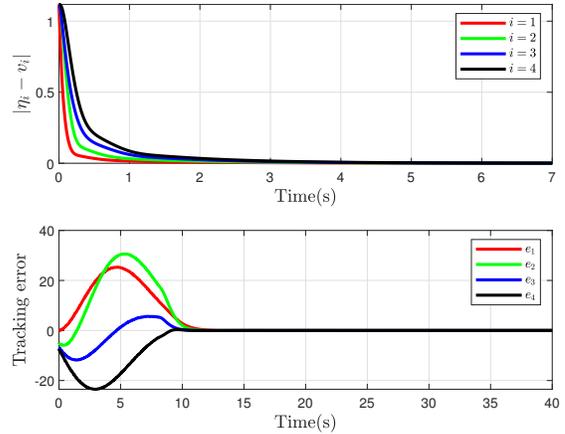}
    \caption{Exogenous signal estimation error and the tracking error under Algorithm \ref{algorithm: ADP}.}
    \label{fig:v_and_trackingerror}
\end{figure}
\begin{figure}[tb!]
    \centering
    \includegraphics[width=0.85\linewidth,trim={1.10cm 1mm 1.10cm 4.50mm},clip]{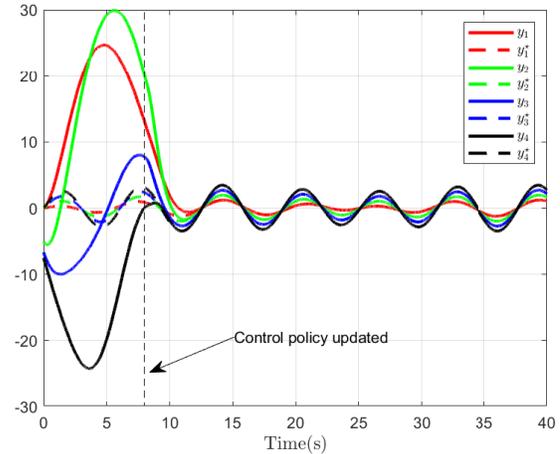}
    \caption{Actual outputs $(y_i)$ and the desired outputs $(y_{i}^*)$ generated under Algorithm \ref{algorithm: ADP}.}
    \label{fig:output}
\end{figure}

\bibliographystyle{IEEEtran}
\bibliography{ref}
\end{document}